\documentclass{llncs}

\usepackage{epsfig}
\usepackage{graphicx}
\usepackage{amsmath}

\begin{document}

\author{Mordecai Golin\inst{1} 
\and John Iacono\inst{2}%
\thanks{Supported in part by NSF grant CCF-1018370 and by an Alfred P.~Sloan fellowship.}
\and Danny Krizanc\inst{3} 
\and Rajeev Raman \inst{4}%
\thanks{Research done while on study leave from the University of Leicester.} 
\and S.~Srinivasa~Rao\inst{5}%
\thanks{Supported by the Seoul National University Foundation Research Expense.}
\and Sunil Shende\inst{6}
}

\institute{
Hong Kong University of Science and Technology
\and
Polytechnic Institute of New York University 
\and
Wesleyan University
\and
University of Leicester
\and
Seoul National University
\and
Rutgers University, Camden}

\title{Encoding 2D Range Maximum Queries}
\maketitle

\begin{abstract}
We consider the \emph{two-dimensional range maximum query (2D-RMQ)} 
problem: given an array $A$ of ordered values, to pre-process
it so that we can find the position of the largest element in  a
(user-specified) range of rows and range of columns.  
We focus on determining the \emph{effective} entropy of 2D-RMQ, i.e., how many bits are
needed to encode $A$ so that 2D-RMQ queries can be answered \emph{without} access to $A$.
We give tight upper and lower bounds on the expected
effective entropy for the case when $A$ contains independent identically-distributed
random values, and new upper and lower bounds for arbitrary $A$, for the case when
$A$ contains few rows.  The latter results improve upon
upper and lower bounds by Brodal et al. (ESA 2010). 
We also give some efficient data structures for 2D-RMQ whose space usage is close to
the effective entropy.
 \end{abstract}

\section{Introduction}
\label{sec:intro}
In this paper, we study the two-dimensional 
\emph{range maximum query} problem (2D-RMQ). 
The input to this problem is a two dimensional~$m$ by $n$ array~$A$ of~$N=m\cdot n$ 
elements from a totally ordered set. We assume w.l.o.g. 
that~\mbox{$m \le n$} and that all the entries of~$A$ are 
distinct (identical entries of~$A$ are ordered lexicographically by their index).
We consider queries of the following types.
A 1-sided query consists of the positions in the array in the range~\mbox{$q=[1\cdots m] \times [1\cdots j]$}, where~\mbox{$1 \le j \le n$}. 
(For the case $m=1$ these may also be referred to as \emph{prefix maximum} queries.)
For a  2-sided query the range is~\mbox{$q=[1\cdots i] \times [1\cdots j]$}, where~\mbox{$1 \le i \le m$} and~$1 \le j \le n$; for a 3-sided query, ~\mbox{$q=[1\cdots i] \times [j_1\cdots j_2]$}, where~\mbox{$1 \le i \le m$} and~$1 \le j_1 \le j_2 \le n$
and for a 4-sided query, the query range is~\mbox{$q=[i_1\cdots i_2] \times [j_1\cdots j_2]$}, where~\mbox{$1 \le i_1 \le i_2 \le m$} and~$1 \le j_1 \le j_2 \le n$.
In each case, the response to a query is 
the position of the maximum element in the
query range, i.e., $\mathrm{RMQ}(A,q) = \mathrm{argmax}_{(i,j)\in q} A[i,j]$.
If the number of sides is not specified we assume the query is 4-sided.

We focus on the space requirements for answering this query in the
\emph{encoding} model \cite{Brodal2010}, where the aim is to pre-process
$A$ and produce a representation of $A$ which allows 2D-RMQ queries to be answered
\emph{without} accessing $A$ any further.  We now briefly motivate this particular question.
Lossless data compression is concerned with the information content of data,
and how effectively to compress/decompress the data so that it uses space close to its
information content.  However, there has been an explosion of interest in operating \emph{directly} (without decompression) on
compressed data via \emph{succinct} or \emph{compressed} data structures 
\cite{DBLP:journals/siamcomp/GrossiV05,DBLP:journals/jacm/FerraginaM05}.  
In such situations, a fundamental 
issue that needs to be considered is the ``information content of the data \emph{structure},''
formalized as follows. 
Given a set of objects $S$, and a set of queries $Q$, consider the equivalence
class ${\cal C}$ on $S$ induced by $Q$, where two objects from $S$ are equivalent if they provide the same answer to all queries in~$Q$.  Whereas traditional succinct data structures
are focussed on storing a given $x \in S$ using at most  $\lceil \log |S|\rceil$
bits\footnote{All logarithms are to base 2 unless stated otherwise.} --- the \emph{entropy} of $S$, we consider 
the problem of storing $x$ in $\lceil \log |{\cal C}| \rceil$ 
bits --- the \emph{effective entropy} of $S$
with respect to $Q$ ---
while still answering queries from $Q$ correctly.  
In what follows, we will abbreviate 
``the {effective entropy} of $S$
with respect to $Q$'' as  ``the effective entropy of $Q$.''

Although this term is new, the question is not: 
a classical result, using \emph{Cartesian trees} \cite{Vuillemin1980},
shows that given an array $A$ with 
$n$ values from $\{1,\ldots,n\}$, only $2n - O(\log n)$ bits are required to answer 
1D-RMQ without access to $A$, as opposed to the $\Theta(n \log n)$ bits needed
to represent $A$ itself.  The low effective entropy of 1D-RMQ 
is useful in many applications, e.g. it is used to simulate access 
to LCP information in compressed suffix arrays (see e.g. \cite{Sadakane2007}). 
This has motivated much
research into data structures whose space usage is close to the 
$2n - O(\log n)$ lower bound and which can 
answer RMQ queries quickly
(see \cite{DBLP:journals/siamcomp/FischerH11} and references therein).  In addition
to being a natural generalization of the 1D-RMQ, 
the 2D-RMQ query is also a standard kind of range 
reporting query.

\paragraph{Previous Work.}
\label{sec:previous_work}
The 2D-RMQ problem, as stated here, was proposed by Amir et al.~\cite{Amir2007}. (The variant where elements are associated with a sparse set of points in 2D, introduced by \protect{\cite{Gabow1984}}, is fundamentally different and is not discussed further here.)
Building on work by Atallah and Yuan~\cite{Atallah2009a},
Brodal et al.  \cite{Brodal2010} 
gave a hybrid data structure that combined a compressed
``index'' of $O(N)$ bits along with the original array~$A$.
Queries were answered using the index along with
$O(1)$ accesses to~$A$. 
They showed that this is an optimal point on the
trade-off between the number of accesses and the memory used.
In contrast, Brodal et al.
refined Demaine et al.'s~\cite{Demaine2009} earlier lower bound to show that
the effective entropy of 2D RMQ is  
$\Omega(N \log m)$ bits, thus resolving in the negative Amir et al.'s open question regarding
the existence of an $O(N)$-bit encoding for the 2D-RMQ problem. 
Brodal et al. also gave an $O(N \min\{m, \log n\})$ bit
encoding of $A$.  Recalling that $m$ is the smaller of the two dimensions,
it is clear that Brodal et al.'s encoding is non-optimal unless
$m = n^{\Omega(1)}$.

\subsubsection{Our Results.} We primarily consider two cases of the above problem: (a) the \emph{random} case, where
the input $A$ comprises $N$ independent, uniform (real) random numbers from $[0,1)$, 
and (b) the case of \emph{fixed} $m$, where $A$ is worst-case, but 
$m$ is taken to be a (fairly small) constant.  
Random inputs are of interest in practical situations 
and provide insights into the lower bounds of \cite{Demaine2009,Brodal2010} ---
for instance, we show that the 2D-RMQ can be encoded in $O(N)$ expected bits as opposed to $\Omega(N \log m)$ bits for the worst case --- that could inform the design of \emph{adaptive}
data structures which could use significantly less space for practical inputs.  For the case of fixed $m$, we determine the precise constants in the effective
entropy for particular values
of $m$ --- applying the techniques of Brodal et al. directly yields significantly non-optimal lower and upper bounds.  These results use ideas that may be relevant to solving the asymptotic version of the problem.
The majority of our effort is directed towards determining the effective entropy and providing concrete encodings that match the effective entropy, 
but we also in some cases provide data structures that support range maximum
queries space- and time-efficiently on the RAM model
with logarithmic word size.

\paragraph{Effective Entropy on Random Inputs.} 
We first consider the 1D-RMQ problem for an array $A[1\cdots n]$
(i.e. $m=1$) and show
that, in contrast to the worst case lower bound of $2n - O(\log n)$ bits,
the expected effective entropy of RMQ is $\le cn$ bits for $c
\approx 1.736\ldots$,  where the precise 
value of $c$ equals $2\sum_{i=1}^{\infty} \frac{\log i}{(i+1)(i+2)}$.  
We also give another encoding that is more ``local'' and achieves
expected $cn + o(n)$ bits for $c \le 1.9183\ldots$.

In the 2D case, $A$ is an $m \times n$ array with $2 \le m \le n$. 
We show bounds on the expected effective entropy of RMQ as below:
\begin{center}
\begin{tabular}{|c|c|c|c|}
\hline
             1-sided             &       2-sided        &        3-sided    &           4-sided\\
\hline
        $\Theta((\log n)^2)$ bits  &  $\Theta((\log n)^2 \log m)$ bits & $\Theta(n (\log m)^2)$ bits  & $\Theta(nm)$ bits\\
\hline
\end{tabular}
\end{center}
The 2D bounds are considerably lower than the known worst-case 
bounds of $O(n \log m)$ for the 1-sided case,
$O(nm)$ for the 2-sided case, and known lower bounds of $\Omega(nm)$ and $\Omega(nm \log m)$ for
the 3-sided and 4-sided cases respectively. 
The above results also hold in the weaker model where we assume all
permutations of $A$ are equally likely.
We also give a data structure that supports 
(4-sided) RMQ queries in $O(1)$ time using expected $O(nm)$ bits of space.  


\paragraph{Effective Entropy for Small $m$.} Our results for the 2D RMQ 
problem (4-sided queries) with worst-case inputs are as follows:
\begin{enumerate}
\item We give an encoding based on ``merging'' Cartesian trees\footnote{This encoding has also been discovered by Brodal (personal communication).}.  While this
encoding uses $\Theta(nm^2)$ bits, the same as that of
Brodal et al. \cite{Brodal2010}, it has lower constant factors: e.g., 
it uses $5n - O(\log n)$ bits when $m=2$ rather than $7n - O(\log n)$ bits \cite{Brodal2010}. We also give a data structure for the case 
$m=2$ that uses $(5+\epsilon) n$ bits and
answers queries in $O(\frac{\log(1/\epsilon)}{\epsilon})$ time for any  $\epsilon > 0$.

\item We give a lower bound on the effective entropy 
based on ``merging'' Cartesian trees. This lower
bound is not aymptotically superior to the
lower bound of Brodal et al. \cite{Brodal2010}, but for all 
fixed $m < 2^{12}$ it gives a better lower bound than that of Brodal et al.
For example, we show that for $m = 2$, the effective entropy is $5n - O(\log n)$ bits,
\emph{exactly} matching the upper bound, but
the method of Brodal et al. yields only a lower bound of $n/2$ bits.

\item For the case $m=3$, we give an encoding that requires $(6 + \log 5)n - O(\log n) \approx 8.32n$
bits\footnote{All logarithms are to the base $2$.}.   
Brodal et al.'s approach requires 
$(12 + \log 5)n - O(\log n) \approx 14.32n$ bits 
 and the method in (1) above would require about $9n$ bits.  
Our lower bound from (2) is $8n - O(\log n)$ for this case.
\end{enumerate}
The paper is organized as follows: in Section~\ref{sec:random} we give bounds on the expected entropy for random inputs, in Section~\ref{sec:smallm} we consider the case of small $m$ and Section~\ref{sec:datastructures} gives the new data structures. 

\section{Random Input}
\label{sec:random}

In this section we consider the case of ``random'' inputs, where the array
$A$ is populated with $N = m \cdot n$ independent uniform random real
numbers in the range $[0,1)$.  We first consider the 1D-RMQ problem 
giving two encodings, one optimal but less convenient to decode
(Theorem~\ref{thm:2sided1D})  and 
another that is less compact, but easier to decode 
(Theorem~\ref{thm:2sided1Db}).  We then consider 
the 2D cases, beginning with 1-sided queries (Theorem~\ref{thm:1sided2D}),
2-sided (Theorem~\ref{thm:2sided2D}) and finally 4- and 3-sided
queries (Theorem~\ref{thm:4sided2D}).

\subsection{1D RMQ problems.}
For the 1D case, we begin by outlining the Cartesian tree 
\cite{Vuillemin1980}.
Given an array $A$ containing $n$ distinct
numbers, its Cartesian tree is 
an \textit{unlabeled} $n$-node 
binary tree in which each node corresponds to a unique position
in the array, and is defined recursively as follows: the root of
the tree corresponds to position $i$, where $A[i]$ is the maximum
element in $A$, and the left and right subtrees of the root are
the Cartesian trees for the sub-arrays $A[1 \dots i-1]$ and
$A[i+1 \dots n]$ respectively (the Cartesian tree of a null array
is the empty binary tree).  A Cartesian tree for an array $A$ can
be used to answer 1D-RMQ on $A$ via a lowest common ancestor
query on the Cartesian tree.  We first show:

\begin{theorem}
\label{thm:2sided1D}
The expected effective entropy of 2-sided queries on a 1D array of size $n$ is at most $cn + O(\log n)$ bits where $c \approx 1.736 \ldots$ has the exact value
\begin{align*}
c & = 2 \sum_{k=1}^{\infty} \frac{\log k}{(k+1)(k+2)}
\end{align*}
\end{theorem}

\begin{proof} 
We first give an encoding of a Cartesian tree.
In what follows, we number the nodes of the Cartesian tree
in in-order, so that node $v$ corresponds to $A[v]$.  Each
node $v$ is the root of a subtree of size $s_v \ge 1$ (which 
represents a sub-array of $A$ of size $s_v$). The
\emph{relative offset} $o_v$ of $v$ is an integer
in the range $0..s_v - 1$ representing its relative
position in this sub-array (if $v$ has a left child $x$,
then $o_v = s_x$, otherwise $o_v = 0$).  
The encoding is obtained by visiting the nodes of the Cartesian tree
in pre-order and writing down the values $o_v$ in the order they
are encountered.  It is not hard to see that the Cartesian tree
can be uniqely decoded from this encoding. If 
$root$ denotes the root of the Cartesian tree,  the first value
in the sequence, which is $o_{root}$, gives the
sizes of the left and (since $n$ is known) right subtrees: we
obtain the sub-sequences corresponding to the subtrees and
recurse.  

A naive implementation of this encoding
requires $\Theta(n \log n)$ bits since each
relative offset potentially needs $\Theta(\log n)$ bits;
even encoding each relative offset in $\lceil \log o_v \rceil$ bits
may cause this encoding to exceed the known upper bound of $2n - O(1)$ bits.
For a more space-efficient encoding, define the
\emph{weight} of a node $v$, $w_v$, as the product of
$s_x$ for all nodes $x$ in $v$'s subtree, including
$v$ itself.  The sequence of values $o_v$, written in pre-order, is viewed
as a mixed-radix integer, where the radix of $o_v$ is 
$s_v$, and the root and last node in pre-order are taken
to be the {least} and most significant digit respectively.
The value of the resulting
integer is clearly in the range $0..w_{root} - 1$.  This
integer is computed as follows.  We traverse the Cartesian
tree in pre-order. Suppose that $v$ is a node with left
child $l$ and right child $r$. Before returning from $v$
to $v$'s parent, having completed the
traversal of $l$ ($r$) of
a node $v$, and encoded the subtrees rooted at $l$
($r$) as integers $e_l$ ($e_r$), we encode $v$ as
the integer $e_v = e_r w_l s_v + e_l s_v +  o_v$ in the
range $0 .. w_v -1$.  If the final encoding is
$e_r$, then ${{e_r}\bmod{n}} = o_r$; the rest of
the decoding can be done essentially by converting
the recursive decoding algorithm described in the previous
paragraph to an iterative one. Observe
that the size of the encoding is $\lceil \log w_{root} \rceil  = 
\left \lceil \sum_v \log s_v \right \rceil$ bits; we now bound the 
latter quantity.

For every $i$ such that $0 \leq i \leq n-1$, an array of
random independent numbers $A$ will have its maximum element at
position $i+1$ with probability $1/n$. Consequently, the
associated distribution of Cartesian trees will witness a tree
with exactly $i$ nodes in its left subtree and $n-1-i$ nodes in
its right subtree with probability $1/n$; the root will
correspond to position $i+1$.  
Let $S(n)$ denote the expected value of
$\sum_v \log s_v$ for a
random array $A$ of size $n$, and note that
$\lceil S(n) \rceil$ is the expected size of the encoding.
Taking $S(0) = 0$, then for all $n \geq 1$:
\begin{align}
S(n)  & = {\log n} + \frac{1}{n} \sum_{i=0}^{n-1} (S(i)
                        + S(n-1-i))  \notag\\
 & =  {\log n} + \frac{2}{n} \sum_{i=0}^{n-1} S(i) \label{eq-entropy}
\end{align}
In fact, Equation (\ref{eq-entropy}) is \textit{identical} to the 
recurrence for the entropy of $n$-node
\textit{random binary search trees} \cite{Kieffer2009},
viz. counting the expected number of bits
needed to describe a binary search tree produced by a random
permutation of $n$ distinct numbers. It has
the \textit{exact} solution $S(0)=0$ and 
\[ S(n) = {\log n} + 2(n+1)\sum_{i=1}^{n-1} \frac{\log
  i}{(i+1)(i+2)} \]
for $n \geq 1$.  The result follows.
\qed
\end{proof}
Although not a primary concern of this section, it should
be noted that the above encoding appears to require
$O(n^2)$ time to decode. A less compact encoding, which is more
``local'' and is linear-time decodable is given below:
\begin{theorem}
\label{thm:2sided1Db}
There is a linear-time decodable encoding of 1D RMQ that uses 
$cn + o(n)$ bits for $c = 1.9183\ldots < 1.92$.
\end{theorem}

\begin{proof}
We study the distribution of different kinds of nodes in the Cartesian tree 
of a random array. 
Each node in a Cartesian tree can be of four types -- it can have 
two children (type-2), 
only a left or right child (type-L/type-R), 
or it can be a leaf (type-0). 
Consider an element $A[i]$ for $1 \le i \le n$ 
and observe that the type of the $i$-th node in the 
Cartesian tree in inorder (which corresponds to $A[i]$) 
is determined by the relative values of $l = A[i-1]$,
$m = A[i]$ and $r = A[i+1]$ (adding dummy random
elements in $A[0]$ and $A[n+1]$). Specifically:
\begin{enumerate}
\item  if $r > m > l$ then node $i$ is type-L;
\item  if $l > m > r$ then node $i$ is type-R;
\item  if $l > r > m$ or $r > l > m$ then node $i$ is type-0 and
\item  if $m > l > r$ or $m > r > l$ then node $i$ is type-2.
\end{enumerate}
In a random array, the probabilities of the alternatives above are
clearly $1/6, 1/6$, $1/3$ and $1/3$.  By linearity of expectation, if
$N_x$ is the random variable that denotes the number of type-$x$ nodes,
we have that $\mbox{\rm E}[N_0] = \mbox{\rm E}[N_2] = n/3$ and
$\mbox{\rm E}[N_L] = \mbox{\rm E}[N_R] = n/6$.  The encoding consists 
in traversing the Cartesian tree in either level-order or in
depth-first order (pre-order) and writing down the label of each node in
the order it is visited: it is known that this suffices to reconstruct
the Cartesian tree \cite{Jacobson89,BenoitDMRRR05}.  The sequence of labels
is encoded using arithmetic coding, choosing the probability of 
type-0 and type-2 to be 1/3 and that of type-L and type-R to be 1/6.
The coded output would be of size $\log 6 (N_R + N_L) +
\log 3 (N_0 + N_2)$ (to within lower-order terms)
\cite{DBLP:reference/algo/HowardV08}; plugging in the expected
values of the random variables $N_x$ gives the result.  It is easy
to see how to decode the tree from this encoding in linear time.
\qed
\end{proof}

\subsection{2D RMQ problems.}

We now consider the 2D case.

\begin{theorem}
\label{thm:1sided2D}
The expected effective entropy of 1-sided queries on an $m \times n$ array
is $\Theta(\log^2 n)$ bits. 
\end{theorem}

\begin{proof}
For the upper bound observe that we can recover the answers to the 1-sided
queries by storing the positions of the prefix maxima, i.e., those positions,
$(i,j)$, such that the value stored at $(i,j)$ 
is the maximum among those in positions
$[1 \cdots m] \times [1 \cdots j]$. 
Since the position $(i,j)$ is a prefix maximum with
probability $1/jm$ and can be stored using 
$\lceil \log (nm+1) \rceil$ bits, the expected number
of bits used is at most 
$\sum_{i=1}^m \sum_{j=1}^n (\lceil \log (nm+1) \rceil)/jm 
= O(\log^2 n)$ bits. 

Consider a random source that generates  $n$ elements of 
$\{ 0, 1, \ldots, m \}$ as follows: 
the
$i$th element of the source is $j$ if the answer to the query
$[1 \cdots m] \times [1 \cdots i]$ is $(j,i)$ for some $j$, and 0
otherwise. Clearly the entropy of this source is a lower
bound on the expected size of the encoding. This source produces $n$
independent elements of 
$\{ 0, 1, \ldots, m \}$ with the $i$th  equal to $j$, $j = 1, \ldots, m$,
 with probability $1/im$ and is equal to 0 with probability $1 - 1/i$.
I.e., its entropy is 
$\sum_{i=1}^n [ (1-1/i) \log (\frac{i+1}{i}) + \sum_{j=1}^m
\log (im) /im] = \Omega( \log^2 n)$ 
bits.
\qed
\end{proof}
\begin{theorem}
\label{thm:2sided2D}
The expected effective entropy of 2-sided queries on an $m \times n$ array
is $\Theta(\log^2 n \log m)$ bits. 
\end{theorem}
\begin{proof}
As in the proof of Theorem \ref{thm:1sided2D}, we store a list of the positions
of the 2-sided prefix maxima sorted by their values. By 2-sided prefix
maxima we mean those positions $(i,j)$ where the value in
that position is maximum among all those in $[1 \cdots i] \times [1 \cdots j]$.
The answer to any query is the position of the largest such 2-sided 
prefix maximum 
inside the query. This can be determined from the sorted list of positions.
The expected number of bits in the encoding is at most
$\sum_{i=1}^m \sum_{j=1}^n \lceil \log (nm+1) \rceil /ij = O(\log^2 n \log m)$ 
bits. 

The lower bound is also similar. From an encoding for 2-sided queries
of an $m \times n$ array,  we can create a 
source of $nm$ independent bits with a bit being 1 if and only if the
answer to the query $[1 \cdots i] \times [1 \cdots j]$ is $(i,j)$ 
(which occurs with probability $1/ij$).
The entropy of this source is at least 
$\sum_{i=1}^m \sum_{j=1}^n \log(ij)/ij = \Omega(\log^2 n \log m)$ bits.
\qed
\end{proof}

\begin{theorem}
\label{thm:4sided2D}
The expected effective entropies of 4-sided and 3-sided queries on an $m \times n$ array
are $\Theta(n m)$ bits and $\Theta(n (\log m)^2)$ respectively. 
\end{theorem}

\begin{proof}
We begin with the 4-sided case\footnote{NB: %
positions are given as row index then column index, not $x$-$y$ coordinates.}.
For each position $(i,j)$ we store a region which has the property
that for any query containing $(i,j)$ and 
lying entirely within that region, $(i,j)$ is the
answer to that query. This contiguous region  is
delimited by a monotone (along columns or rows) sequence of positions in
each of the quadrants defined by $(i,j)$. A position $(k,l)$ delimits the 
boundary of the region of $(i,j)$ if the value in position $(k,l)$ is
the largest and the 
value in position $(i,j)$ is the second largest, 
in the sub-array defined by $(i,j)$ and $(k,l)$, i.e.,
any query in this sub-array not including positions on row $k$ or column $l$
is answered with $(i,j)$ (dealing with boundary conditions appropriately).
The answer to any query is the (unique) position inside the query whose region 
entirely contains the query.
For each position, we store a clockwise ordered list of the positions delimiting
the region (starting with the position above $i$ in column $j$)
by giving the position's column and row offset from $(i,j)$.

\begin{figure}
\begin{center}
\includegraphics[height=0.5in]{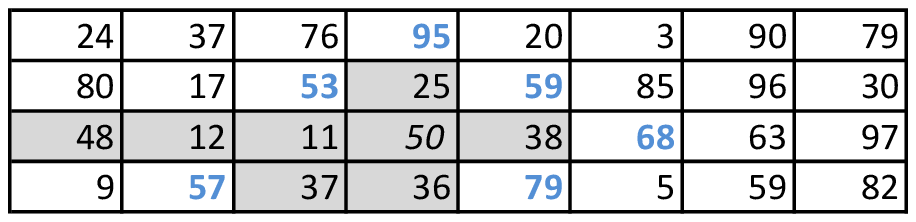}~~\includegraphics[height=0.5in]{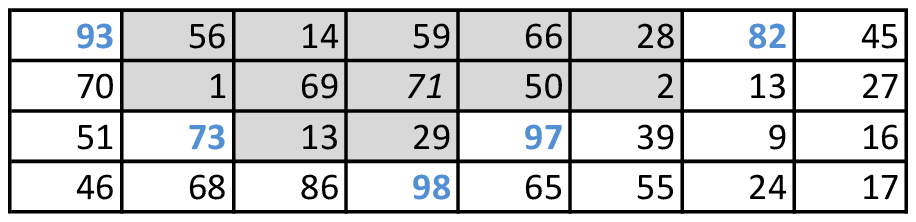}
\end{center}
\caption{Regions for the italicised values for 4-sided queries (left) and 3-sided queries (right), with region delimiters in blue.}
\label{fig:delimit}
\end{figure}

The expected number of bits 
required to store any region is at most 4 times
the expected number bits required to store the region of $(1,1)$, whose
boundary runs diagonally from the first column to the first row. A position $(i,j)$
delimits the boundary of $(1,1)$ ($i=1, \ldots, m+1$, $j=1, \ldots, n+1$ 
excluding the case $i=j=1$) with probability 
$1/ij(ij-1)$ (if it contains the largest
value and $(1,1)$ is the second largest in the sub-array $[1 \cdots i] \times
[1 \cdots j]$) and its offsets require at most 
$2 (\lceil \log(i+1) \rceil + \lceil \log(j+1) \rceil)$ 
bits to store. I.e., the expected number of bits stored per position is
at most 
$$
4 \cdot 
\left(
\sum_{i=2}^{m+1} 
\frac{2 \lceil \log (i+1) \rceil}{i(i-1)} + 
\sum_{j=2}^{n+1} \sum_{i=1}^{m+1} 
\frac{2 (\lceil \log(i+1) \rceil + \lceil \log(j+1) \rceil)}{ij(ij-1)} 
\right)
= O(1).
$$
By linearity of expectation, the expected number of bits stored is $O(nm)$. 
The bound is tight as we can generate $\lceil n/2 \rceil m$ equiprobable independent
random bits (of entropy $\Omega(nm)$) from $A$ by reporting 1 iff the 
answer to  the query consisting of the two positions $(2i-1,j)$ and $(2i,j)$ 
is $(2i,j)$, for $i = 1, \ldots, \lceil n/2 \rceil$, 
$j = 1, \ldots, m$. 

For the 3-sided case, recall that we focus on queries that are open to the ``top'' side.
We again define the region of 
position $(i,j)$ as the area such that $(i,j)$ is the answer to
any query containing $(i,j)$ and lying entirely within that area.  
In contrast to the 4-sided case, the region of a point 
may be empty (if it is not a prefix maximum in its column).  For points
with non-empty regions, their region is delimited on the right by 
a monotone sequence of positions $(k,l)$ such that $l > j$ for
all positions, and $k > i$ for all positions but one (see Fig.~\ref{fig:delimit}).  
The left delimiters  are symmetric, and the region is 
obviouly delimited from below by the next prefix maximum in the $j$-th column.
To answer RMQs, we store all (ordered) pairs $(p,q)$ such that 
position $q$ delimits the region of position $p$ (assuming $p$ has a non-empty region).
The pairs are stored sorted by $p$'s column, and are represented as follows (all numbers are
assumed stored in a self-delimiting manner, e.g. using Gamma codes). We store
$p$ and $q$'s row number using $O(\log m)$ bits and the difference
between  the columns of $p$ and $q$ using $O(1 + \log (|j - l| + 1))$ bits.

The pair $(p,q)$ is stored iff the value in position $q = (k,l)$ is the largest and the 
value in position $p = (i,j)$ is the second largest
in the sub-array defined by $p$ and $q$, i.e.  
$A[1\cdots \max\{i,k\}][j \cdots l]$ (assuming
 that $l > j$).  For a fixed pair $(p,q), p\ne q$, the probability that $(p,q)$ is stored
is $\frac{1}{(ab)(ab-1)} \le 2/(ab)^2$, where $a = \max\{i,k\}$ and $b = |j-l|+1$.  
We now calculate the expected cost of storing $(p,q)$ over all pairs $(p,q)$, taking
$j = 1$ to simplify the summation (arbitrary $j$ will have a summation at most
twice that of $j=1$).  
\begin{eqnarray*}
2 \cdot \sum_{i = 1}^m \sum_{l = 1}^n 
\sum_{k = 1}^m \frac{1}{\max\{i,k\}^2 l^2} \cdot O(\log m + \log l) 
& = &  O \left ( \sum_{i = 1}^m \sum_{l = 1}^n \frac{(\log m + \log l)}{i l^2} \right )\\
~=~ O \left ( \sum_{i = 1}^m (\log m) / i \right )
& = &  O((\log m)^2).
\end{eqnarray*}
The summation uses the observation that, for any fixed $i$,
$\sum_{k = 1}^m \frac{1}{(\max\{k,i\})^2} =  \sum_{k = 1}^i \frac{1}{i^2} +
\sum_{k = i+1}^m \frac{1}{k^2} = O(1/i)$.
%
%
Summing over all $j$, we get that the expected effective entropy is $O(n (\log m)^2)$.
For the lower bound, the $n$ columns can be considered independent $1 \times m$ prefix maxima
problems each requiring expected $\Omega( \log^2 m )$ bits by 
Theorem \ref{thm:1sided2D}.\qed
\end{proof}

\section{Small $m$}
\label{sec:smallm}

Brodal et al. \cite{Brodal2010} gave a 2D-RMQ encoding of size essentially
$\left ( \frac{m(m+3)}{2} - O(\log m) \right ) 
\cdot 2n \approx n \cdot m(m+3)$ bits for a $m \times n$ array. In order that
precise comparisons can be made for fixed values of $m$, we outline their approach.
For each of the $m$ rows of the matrix, they store a Cartesian tree for
that row, and for each of the $(m)(m-1)/2$ possible subranges of rows,
they store a Cartesian tree for for the maximum value in each column that lies within
that set of rows.  Since we consider $m$ fixed in this section, the space bound 
for these Cartesian trees is $((m)(m+1)/2)(2n - O(\log n))$ bits which is essentially
$2n \cdot (m)(m+1)$ bits.  Given any query spanning a subrange of rows, the 
Cartesian tree for that subrange tells us which column the range maximum lies in.
However, to find which row the maximum lies, in Brodal et al. also store a Cartesian
tree for each column of the matrix.  The space used by these column-wise Cartesian trees needs to
be calculated more carefully since $m$ is small, and is taken to be
$\left \lceil n \cdot \log \left ( \frac{1}{m+1} {{2m} \choose {m}} \right ) \right \rceil$ (we do not take the
ceiling of the log since the Cartesian trees for all columns could be encoded together). Specifically
this gives:
\begin{center}
\begin{tabular}{|c|c|c|c|}
\hline $m$ & row-wise CT & column-wise CT & Total \\
\hline 
     2   &  $6n - O(\log n)$  & $n$  & $7n - O(\log n)$ \\
\hline
     3   &  $12n - O(\log n)$ & $n \cdot \log 5$  & $\approx 14.32n$ \\
\hline
     4   &  $20n - O(\log n)$ & $n \cdot \log 14$ & $\approx 23.81n$\\
\hline
\end{tabular}
\end{center}
Furthermore, Brodal et al. showed that the effective entropy of 2D-RMQ is
at least $\log ( (\frac{m}{2}!)^{\lfloor \frac{n-m/2+1}{2} \rfloor} )$ bits.
For $m=2$, their techniques give a lower bound of $n/2$, but this is worse than
the obvious lower bound of $4n - O(\log n)$ obtained by considering each row
independently.

In this section we improve upon these results for small $m$. Our main tool
is the following lemma:
\begin{lemma}
\label{lem:merging}
Let $A$ be an arbitrary $m \times n$ array, $m\ge 2$. 
Given an encoding capable of answering
range maximum queries  of the 
form $[1 \cdots (m-1)] \times [j_1 \cdots j_2]$ ($1 \leq j_1 \leq j_2 \leq n$) 
and an encoding answering 
range maximum queries on the last row of $A$, 
$n$ additional bits are necessary and sufficient to
construct an encoding answering queries of the form
$[1 \cdots m] \times [j_1 \cdots j_2]$ ($1\leq j_1 \leq j_2 \leq n$) on $A$. 
\end{lemma}

\begin{proof}
The proof has two parts, one showing sufficiency (upper bound) and the
other necessity (lower bound).  

\paragraph{Upper Bound.} We construct a \emph{joint} Cartesian tree that can be used
in answering queries of the form $[1 \cdots m] \times [j_1 \cdots j_2]$ for 
$1 \leq j_1 \leq j_2 \leq n$, using an additional $n$ bits.  The root of the 
joint Cartesian tree is either the answer to the query  $[1 \cdots (m-1)] \times
[1 \cdots n]$ or $[m] \times [1 \cdots n]$. We store a single bit indicating the larger
of these two values. We now recurse on the portions of the array to the left and right
of the column with the maximum storing a single bit, which indicates which sub-problem
the winner comes from, at each level of
the recursion. Following this procedure, using the $n$ additional bits it
created, we can construct the joint Cartesian tree.  To answer queries of
the form $[1 \cdots m] \times [i \cdots j]$, the lowest common
ancestor $x$ (in the joint Cartesian tree) of $i$ and $j$ gives us the
column in which the maximum lies.  However, the comparison that placed 
$x$ at the root of its subtree
also tells us if the maximum lies in the $m$-th row or in rows $1\cdots m-1$; in
the latter case, query the given data structure for rows $1\cdots m-1$.

\paragraph{Lower Bound.}
For simplicity we consider the case $m=2$ --- it is easy to see that
by considering the maxima of the first $m-1$ elements of each column the general
problem can be reduced to that of an array with two rows. 
Let the elements of the top and bottom rows be $t_1, t_2, \ldots , t_n$ 
and $b_1, b_2, \ldots, b_n$.  Given two arbitrary Cartesian 
trees $T$ and $B$ that describe the answers to the top and bottom rows, the  
procedure described in the upper bound for constructing the Cartesian 
tree for the $2 \times n$ array from $T$ and $B$ makes exactly $n$ 
comparisons between some $t_i$ and $b_j$.  
Let $c_1,\ldots,c_n$ be a bit string that describes the outcomes of these
comparisons in the order which they are made. We now show how to
assign values to the top and bottom rows that are consistent with any given
$T$, $B$, and comparison string $c_1,\ldots,c_n$. Notice this is different from
(and stronger than) the trivial observation that there exists a $2\times n$
array $A$ such that merging $T$ and $B$ must use $n$ comparisons: we show
that $T$, $B$ and the $n$ bits to merge the two rows are independent
components of the $2\times n$ problem.

If the $i$-th comparison compares the
maximum in $t_{l_i},\ldots,t_{r_i}$ with the maximum in $b_{l_i},\ldots,b_{r_i}$, 
say that the {range} $[l_i,r_i]$ is \emph{associated} with the
$i$-th comparison. The following properties of ranges follow 
directly from the algorithm for constructing the Cartesian
tree for both rows:
\begin{itemize}
\item[(a)] for a fixed $T$ and $B$, $[l_i,r_i]$ is uniquely
determined by $c_1,\ldots,c_{i-1}$; 
\item[(b)] if $j > i$ then the range associated with $j$ is either 
contained in the range associated with $i$, or is disjoint from the
range associated with $i$.
\end{itemize}

By (a), given distinct bit strings $c_1,\ldots,c_n$ 
and $c'_1,\ldots,c'_n$ that differ for the first time in 
position $i$, the $i$-th comparison would be associated with
the same interval $[l_i,r_i]$ in both cases, and the query $[1..2]\times[l_i..r_i]$ 
then gives different answers for the two bit strings.  Thus, each bitstring
gives distinguishable $2\times n$ arrays, and we now show that each bitstring
gives valid $2\times n$ arrays.

First note that if the $i$-th bit in a given bit string 
$c_1,\ldots,c_n$ is associated with the interval $[l,r]$,
then it enforces the condition that $t_j > b_k$, where
$j = \mathrm{argmax}_{i \in [l, r]} \{t_i\}$ and 
$k = \mathrm{argmax}_{i \in [l,  r]} \{b_i\}$
or vice-versa. Construct a digraph $G$ with vertex set 
$\{t_1,\ldots,t_n\} \cup \{b_1,\ldots,b_n\}$ which contains
all edges in $T$ and $B$, as well as edges for conditions $t_j > b_k$ (or vice
versa) enforced by the bit string. All arcs are directed from the larger value 
to the smaller. 
We show that $G$ is a DAG and therefore there is a partial order of 
the elements satisfying $T$ and $B$ as well as the constraints enforced by
the bit string.

Suppose that for some value of $c_1,\ldots,c_n$, $G$
is not a DAG. 
Pick any cycle in $G$: 
there must be some node $t \in \{t_1,\ldots,t_n\}$ that
is explicitly enforced (i.e. by a comparison) to be greater than some 
$b \in \{b_1,\ldots,b_n\}$, such that some descendant $b'$ of
$b$ in $B$ has been explicitly enforced to be greater 
than an ancestor $t'$ of $t$ (or the symmetric case with $T$ and $B$ 
interchanged must hold). Let the interval associated with
the $b$-$t$ comparison be $[l,r]$.  First consider the case that $b = b'$. 
Since an element that wins a comparison is never compared again, 
the comparison between $b$ and $t'$ must have occurred after the
comparison with $t$, in which case the interval associated with the
$b$-$t'$ comparison is a sub-interval of $[l,r]$ by property (b). 
This means that $t'$ must be a descendant of $t$.  Therefore $b\ne b'$ and
$b'$ must be a proper descendant of $b$. If $b'$ belongs to
$[l,r]$, it will never have been compared prior to the $b$-$t$ 
comparison, and will subsequently only be compared (if at all)
to a descendant of $t$.  If $b'$ does not belong to $[l,r]$, then
there must have been a comparison between $b$, or one of $b$'s ancestors
in $B$, that was won by a proper ancestor $t''$ of $t$, such that the range
$[l'',r'']$ associated with that comparison was split into two parts,
one containing $[l,r]$ and one containing $b$.  Clearly, $b$ could
not have been compared prior to this comparison, and subsequently,
can only be compared to elements from $T$ that are in a different
subtree of $T$ than $t$.
\qed{}\end{proof}
Using Lemma \ref{lem:merging} we show by induction:

\begin{theorem}
\label{thm:smallm_upperbnd}
There exists an encoding solving the 2D-RMQ problem on 
a $m \times n$ array requiring at most 
$n \cdot \frac{m(m+3)}{2}$ bits.
\end{theorem}
\begin{proof}
The theorem follows by induction from Lemma \ref{lem:merging} and the fact that
$2n$ bits are sufficient to store a Cartesian tree of a $1 \times n$ array (the base case). Given an
encoding solving the RMQ problem for a $(m-1) \times n$ array (using 
$(m-1)(m+2)n/2$ bits by induction) and a Cartesian tree
for a 1D array (using $2n$ bits) we construct a solution to the 2D-RMQ problem
on a $m \times n$ array 
combining the two (with the 1D array as the last row of the combined array) by using
Lemma \ref{lem:merging} to construct $m-1$ Cartesian trees answering queries of the form
$[i \cdots m] \times [j_1 \cdots j_2]$ for $1 \leq i \leq m-1$ using $(m-1)n$ additional
bits. 
\qed
\end{proof}

\begin{theorem}
\label{thm:smallm_lowerbnd}
The minimum space required for any encoding for the 2D-RMQ problem
on a $m \times n$ array is at least $n \cdot (3m-1) - O(m \log n)$ bits.
\end{theorem}
\begin{proof}s
The result follows by induction from Lemma \ref{lem:merging}  and the fact that
$2n - O( \log n)$ bits are required to solve the RMQ problem for a $1 \times n$ array 
(the base case). Note that any encoding that solves the 2D RMQ problem for a
$m \times n$ array must be able to solve the 1D RMQ problem on its last row,
the 2D RMQ problem on the array consisting of the first $m-1$ rows as well 
as queries of the form $[1 \cdots m] \times [j_1 \cdots j_2]$ ($1 \leq j_1 \leq j_2 \leq n$).
The first two problems are entirely independent, i.e., answers to queries from one 
provide no information about the answers to the other. 
The 1D problem requires $2n - O(\log n)$ bits and the additional queries require
at least $n$ bits by Lemma \ref{lem:merging}, i.e., $3n - O(\log n)$ bits are needed on top
of those required for solving the problem on the first $m-1$ rows. 
\qed
\end{proof}
For fixed $m < 2^{12}$ this lower bound is better (in $n$)
than that of Brodal et al. \cite{Brodal2010}.
For the case $m = 2$ our bounds are tight:

\begin{corollary}
\label{cor:m_equals_2}
$5n - O( \log n)$ bits are necessary and sufficient for an encoding 
answering range maximum queries on a $2 \times n$ array.
\end{corollary}

For the case $m=3$, our upper bound is $9n$ bits and our lower bound is
$8n - O(\log n)$ bits. 
We can improve the upper bound slighty:

\begin{theorem}
\label{thm:m_equals_3}
The 2D-RMQ problem can be solved using at most $(6+ \log 5)n + o(n)\approx 8.322 n$ 
bits
on a $3 \times n$ array.
\end{theorem}

\begin{proof}
We refer to the three rows of the array as T (top), M (middle) and B (bottom). 
We store Cartesian trees for each of the three rows (using $6n$ bits). 
We now show how to construct
data structures for answering queries for 
TM (the array consisting of the top and middle rows),
MB (the middle and bottom rows) and TMB (all three rows) 
using an additional $n \log 5 + o(n)$ bits. 
Let $0 \leq x \leq 1$ 
be the fraction of nodes in the Cartesian tree for TMB such that  
the maximum lies in the middle row. 
Given the trees for each row, and a
sequence indicating for each node 
in the Cartesian tree for TMB in in-order,  which 
row contains the maximum for that node,
we can construct a data structure for TMB.
The sequence of row maxima
is coded using arithmetic coding, taking $\Pr[\mathrm{M}] = x$
and $\Pr[\mathrm{B}] = \Pr[\mathrm{T}] = (1 - x)/2$; the output
takes $(-x \log x - (1-x) \log ((1-x)/2))n + o(n)$ bits 
\cite{DBLP:reference/algo/HowardV08}.

We now apply the same procedure as in Lemma \ref{lem:merging} to construct 
the Cartesian trees for TM and MB, storing a bit to answer
whether the maximum is in the top or middle row for TM (middle or
bottom row for MB) for each query made in
the construction of the tree starting with the root. However,
before comparing the maxima in T and M in some range, we check 
the answer in TMB for that range; if TMB reports the answer is
in either T or M, we do not need to store a bit for that range for TM.  It is
easy to see that every maximum in TMB that comes from T or B saves one 
bit in TM or MB, and every maximum in TMB that comes from M saves one 
bit in both TM and MB.  Thus, a total of $2n - (1-x)n - 2xn = (1-x)n$ 
bits are needed for TM and MB.
The total number of bits needed for all three trees (excluding the $o(n)$ term) 
is $(2(1-x) - x \log x - (1-x) \log (1-x))n$. This takes a maximum at $x=1/5$ of $n \log 5$.
\qed
\end{proof}

\section{Data Structures for 2D-RMQ}
\label{sec:datastructures}

\newcommand{\Rank}{\mbox{\sf rank}}
\newcommand{\Select}{\mbox{\sf select}}
\newcommand{\Rankzero}{\mbox{$\mbox{\sf rank}_0$}}
\newcommand{\Rankone}{\mbox{$\mbox{\sf rank}_1$}}
\newcommand{\Selectzero}{\mbox{$\mbox{\sf select}_0$}}
\newcommand{\Selectone}{\mbox{$\mbox{\sf select}_1$}}

In this section, we give efficient data structures for 2D-RMQ. 
We begin by a recap of $\Rank$ and $\Select$ operations on bit strings.
Given a bit string $S[1..t]$ of length $t$, define the following
operations, for $x \in \{0,1\}$:
\begin{itemize}
\item $\Rank_x(S,i)$ returns the number of occurrences of $x$ in the
  prefix $S[1..i]$.
\item $\Select_x(S,i)$ returns the position of the $i$th occurrence of
  $x$ in $S$.
\end{itemize}
Such a data structure is called a \emph{fully indexable dictionary (FID)}
by Raman et al. \cite{DBLP:conf/soda/RamanRR02}, who show that:
\begin{theorem} \label{thm:fid}
There is a FID for a bit string $S$ of size $t$
using at most ${\lg {t \choose r}}+ o(t)$ bits, that
supports all operations in $O(1)$ time 
on the RAM model with wordsize $O(\lg t)$ bits, where
$r$ is the number of 1s in the bit string.
\end{theorem}
Since $\lg {t \choose r} \le t$, the space used by the FID is always $t + o(t)$ bits.

\begin{theorem}
\label{thm:random}
There is a data structure for 2D-RMQ on a
random $m \times n$ array $A$ which answers queries in $O(1)$ time
using $O(mn)$ expected bits of storage.
\end{theorem}
\begin{proof}
Take $N = mn$ and $\lambda = \lceil 2 \log \log N\rceil + 1$, and define
the \emph{label} of an element $z = A[i,j]$ as $\min\{\lceil \log(1/(1-z))\rceil,\lambda\}$ if
$z\ne 0$.
The labels bucket the elements into $\lambda$ buckets of exponentially
decreasing width.  We store the following:
\begin{itemize}
\item[(a)] An $m \times n$ array $L$, where $L[i,j]$ stores the label of
$A[i,j]$.
\item[(b)] For labels $x = 1, 2, \ldots, \lambda -1$, take $r = 2^{2x}$ and 
 partition $A$ into $r\times r$ submatrices 
(called \emph{grid boxes} or \emph{grid sub-boxes} below) using a 
regular grid with lines $r$ apart.
Partition $A$ four times, with the origin of the grid once each at 
$(0,0), (0,r/2),(r/2,0)$ and $(r/2,r/2)$. For each grid box,
and for all elements labelled $x$ in it, store their relative ranks within
the grid box.  
\item[(c)] For all values with label $\lambda$, store their global ranks in the entire array.
\end{itemize}
The query algorithm is as follows:
\begin{enumerate}
\item Find an element with the largest label in the query rectangle.

\item If the query contains an element with label $\lambda$, or if the maximum label
is $x$ and the query fits into a grid box at the granularity associated with label $x$,
then use (c) or (b) respectively to answer the query. 
\end{enumerate}
A query fails if the 
maximal label in the query rectangle is $x < \lambda$ 
but the rectangle does not fit into 
any grid box associated with label $x$.  For this case:
\begin{itemize}
\item[(d)] Explicitly store the answer for all queries that fail.
\end{itemize}
We now give an efficient implementation of the above, and prove the stated
space bounds.

The data structures to support steps (1) and (2) also store (a) and (b) in $O(N)$ bits. 
Firstly, $L$ is represented as a bit string that represents the concatenation of
all elements of $L$ in row-major order, with 
a $x$ encoded in unary as $1^x 0$.
Since the expected number of nodes with label $\ge x$ is $O(N/2^x)$, it follows
that the expected size of this encoding is $O(N)$ bits.  To access $L[i,j]$ in $O(1)$ time,
we store this bit string as an FID (Theorem~\ref{thm:fid}) and 
use the $\Selectzero$ operation on this bit string.  Furthermore, we 
use Brodal et al.'s  ``hybrid'' {2D-RMQ indexing} structure \cite[Theorem 3]{Brodal2010} 
over the array $L$: this data structure stores an index of $O(N)$ bits, 
and answers queries in $O(1)$ time,  
using $O(1)$ comparisons between elements of $L$. The comparisons are implemented
by accessing $L$ and breaking ties arbitrarily, implementing step (1).

However, we cannot use the above approach to implement step (2), since the problem
now focusses on a sparse set of points with label $x$.  Hence, we use
a data structure to solve the following problem: for each label value $x < \lambda$,
answer range maximum queries which lie entirely in a grid box of size 
$r\times r$ where $r = 2^{2x}$.  We allow the data structure for a given grid
box to use $O((t+1)x)$ bits of space, where $t$ is the number of elements in the
box that have label $x$.  Note, however, that $x \le 2 \log \log N + 1$ so 
$r = O((\log N)^4)$.  For any grid box where $r^2 \le c \log N / \log \log N$ for
some sufficiently small constant $c > 0$, this can be done by table lookup, 
since we can write down all the coordinates as well as the relative priorities 
in as a bit string of fewer than $(\log N)/2$ bits: the required table will be
of size at most $O(\sqrt{N})$ words, or $o(N)$ bits.  The space used per grid block
is clearly $O(tx)$ since the bit string comprises just the positions of the points
and their relative priorities, all of which take $O(x)$ bits.  The expected space
used across all grid blocks for label $x$ is at most $O( x N / 2^x)$ bits, summing
up to $O(N)$ expected bits overall.

For larger values of $x$ we use a data structure that takes $o(N / \log \log N)$ expected bits for each value of $x$, but since there are $O(\log \log N)$ values of $x$ this is still $o(N)$ bits overall. We divide each grid box into grid sub-boxes each of side $r'$, where $r' = 2^{2y}$ for the largest $y$ such that $(r')^2 \le c \log N/\log \log N$.  The number of sub-boxes is $N/(r')^2$, or $O(N \log \log N / \log N)$.  For this grid box, we store a matrix $R$ which is $(r/r')\times r$ where each entry corresponds to a row of a sub-box, and contains the largest relative rank in this row of this sub-box.  The space used by $R$ is $O((N/r') \log \log N) = o(N/\log \log N)$. Using Brodal et al's data structure we can do 2D-RMQ queries on $R$.  We also create a $r \times (r/r')$ matrix C where each entry corresponds 
to a column of a sub-box, and similarly we can do 2D-RMQ queries on $C$.
A general query can either be decomposed into $O(1)$ queries on $R$ and $C$, 
plus $O(1)$ 2-sided or 3-sided queries each on one sub-block, or is a 4-sided 
query completely contained in a sub-block, each of which is done by table lookup.  
This implements step (2) in $O(1)$ time using $O(N)$ bits.
A very similar data structure is used to handle elements with label $\lambda$.  We divide the input matrix $A$ into sub-boxes of size $\log N \times \log N$, and create matrices $R'$ and $C'$ which represent elements with label $\lambda$ in each row/column respectively.  We store $R'$ and $C'$ explicitly using $O(N)$ bits each and store a 2D-RMQ indexing structure on $R'$ and $C'$.  To answer queries inside a sub-block, we use an algorithm quite similar to that used for the smaller labels.

We finally need to bound the space usage in (d), and also to describe how to
represent the solutions.  We classify queries based upon their \emph{area}, i.e.
the number of positions they contain. For a given value of the area 
$A$, there are at most $A$ (in fact, considerably fewer) different 
aspect ratios that give rise to that area, and at most $N$ queries with
that aspect ratio, or $NA$ queries in all.  To encode the maximum in
a query with area $A$ requires $O(\log A)$ bits.  The smallest grid granularity
that will contain all queries of area $A$ is the granularity associated with 
label $x = \lceil (\log A)/2 \rceil$.  If a query of area $A$ contains no
positions with labels $\ge x$, it may fail.  The  probability of this
happening is at most $(1-2^{-x-1})^A = 2^{-\Omega(\sqrt{A})}$, so the expected
number of failing queries of area $A$ is $O(N/A^3)$.  
For each area $A \le (\log N)^4$, we store a minimal perfect hash function from 
$[1\cdots N] \times [1 \cdots A] \rightarrow [1..N_A]$, where
$N_A$ is the number of failing queries of area $A$ (the domain of the
hash function specifies the top left corner and, say, the width of the query).
Such hash functions can be stored in $O(N_A + \log \log N)$ bits and
evaluated in $O(1)$ time \cite{DBLP:journals/siamcomp/SchmidtS90}, and are used to index into
an array of length $N_A$ that contains the answer to that query.  The space
used is $\sum_{A=1}^{(\log N)^4} O(N_A \cdot \log A + \log \log N)$ bits, and
since $\mathrm{E}(N_A) = O(N/A^3)$, by linearity of expectation, the expected
space used is $O(N)$. \qed{} \end{proof}
We now show how to support 2D-RMQ queries efficiently on a $2 \times n$ 
array, using space close to that of Corollary~\ref{cor:m_equals_2}.
In the following data structure, we use a Cartesian tree augmented with 
additional leaves, which we call as an {\em augmented Cartesian tree} or
ACT for short. The ACT of a given array $A$ is the Cartesian tree in which
every node is augmented with a leaf in between its left and right children. 
If a node has only a left child then we add the leaf as its second child, and 
if it only has a right child, then we add the leaf as its first child. Finally, if a 
node (in the Cartesian tree) is a leaf, then we add the new leaf as its child.
This structure was used by Sadakane~\cite{Sadakane02-ACT} to 
obtain a space-efficient data structure supporting RMQ queries. The indices
of the array $A$ correspond to the inorder numbers of the nodes in the 
Cartesian tree, and to the preorder numbers of the leaves in the ACT.

\begin{theorem}
\label{thm:m2}
There is a data structure for 2D-RMQ on 
an arbitrary $2 \times n$ array $A$, 
which answers queries in $O(k)$ time
using $5n + O(n \log k / k)$ bits of space,
for any parameter  $k = (\log n)^{O(1)}$.
\end{theorem}

\begin{proof}
For an integer array of length $n$, the {\em 2D-maxHeap}\footnote{Fischer 
and Heun~\cite{DBLP:journals/siamcomp/FischerH11} use 
the term 2D-minHeap as they consider the problem of answering
range minimum queries} of Fischer and 
Heun~\cite{DBLP:journals/siamcomp/FischerH11} 
uses $2n+O(n \log\log n /\log n)$ 
bits and supports RMQ queries on the array in $O(1)$ time. 
The lower order term can be further reduced to $O(n/(\log n)^{O(1)})$
using the tree representation of Sadakane and Navarro~\cite{SN10}. 
We use this 2D-maxHeap structure (which is essentially a 
space- and query-efficient representation of a Cartesian tree) 
to support queries on each of the individual rows, using a total of 
$4n + o(n)$ bits.
The upper bound of Lemma~\ref{lem:merging} shows how to
combine these two Cartesian trees (2D-maxHeaps) using $n$
bits, to answer queries involving both the rows. We refer to these
$n$ bits as the bit vector $M$ (that merges the two Cartesian trees).
Each of these $n$ bits in $M$ corresponds to a unique column in $A$, 
and the bits in $M$ are written in the order of the inorder numbers of 
the nodes in the joint Cartesian tree. 
These bits can be decoded in that order in $O(1)$ time per bit, using 
the Cartesian trees for the individual rows, to reconstruct the joint 
Cartesian tree. A query involving both the rows can be answered by 
decoding the first bit that falls within the query range. Thus the 
worst-case query complexity is $O(n)$ for this representation.

If we represent the joint Cartesian tree for both the rows as a 
2D-maxHeap, and an additional $n$ bits to indicate the column 
maxima, we can answer any RMQ query in $O(1)$ time using 
a total of $7n+o(n)$ bits. We now describe how to reduce the space 
to achieve the trade-off described in the statement of the theorem.

The main idea is to represent the ACT of the joint Cartesian tree 
using a succinct tree representation based on tree decomposition. 
The representation decomposes the ACT into $O(n/k)$ microtrees, 
each of size at most $k$, and represents the microtrees using a 
total of $4n+o(n)$ bits (as the ACT has $2n$ nodes), and stores 
several auxiliary structures of total size $O(n \log k / k)$ bits. It 
supports various queries (such as LCA) in constant time by accessing 
a constant number of microtrees and reading a constant number of 
words from the auxiliary structures. Instead of storing the representations 
of microtrees, we show how reconstruct any microtree in $O(k)$
time (in fact, time proportional to its size) using the bit vector $M$ 
(together with additional auxiliary structures of size $O(n \log k / k)$ bits). 
The new representation consists of the 2D-maxHeap for both the
rows ($4n + o(n)$ bits) and the bit vector $M$ that `merges' these two trees 
($n$ bits) in addition to these auxiliary structures. Thus the 
representation uses $5n + O(n \log k / k)$ bits overall, 
and supports RMQ queries in $O(k)$ time.
We now describe this in detail.

We take the ACT of the joint Cartesian tree and partition it into 
$O(n/k)$ microtrees, each of size at most $k$, for some 
paramater $k \ge 2$, using the {\em tree decomposition} algorithm of 
Farzan and Munro~\cite{FM08}. The microtrees produced by the
decomposition algorithm have the property that for each microtree,
there is at most one node such that one of its children is the root
of another microtree.
In addition, two microtrees can share
a common root.  We modify the decomposition so that whenever 
a node $x$ is the root of two microtrees (a node cannot be the root of 
more than two microtrees, as we have a ternary tree in which one 
of the children of every internal node is a leaf), we remove the node 
$x$ from both the microtrees, and make another microtree containing 
$x$ and its second child (which is a leaf). One can show that the 
number of microtrees produced by the modified decomposition is 
still $O(n/k)$. Now, no two microtrees share a node, and thus we 
obtain a {\em partition} of the ACT into microtrees.

Each leaf in the ACT corresponds to a column in the array $A$, and
as mentioned earlier, the leaves of the ACT in preorder correspond to 
the columns of the array $A$ from left to right. The above partitioning 
procedure splits the ACT such that the columns corresponding the
all the nodes in a microtree are in at most two consecutive {\em chunks} 
in the array $A$ (as all the nodes in a microtree have consecutive 
preorder numbers, except when a node has a child outside the 
microtree; and since there is at most one such node, the claim follows).

We store the auxiliary structures to support LCA and rank/select on
leaves in preorder, using $O(n \log k /k)$ bits. As the ACT has $2n$
nodes, we need $4n+o(n)$ bits to store the representations of all the
microtrees. Instead of storing the representations of the microtrees, 
we show how to reconstruct the representation of 
any microtree in time proportional to its size by simply storing the bit
vector $M$. In addition, we also store auxiliary structures to represent 
the correspondence between the microtrees and the positions in the 
array, as explained below.

We construct two bit vectors $B_1$ and $B_2$ of length $n+O(n/k)$
as follows: we initialize both arrays with $n$ zeroes each. We insert a
$1$ after the $i$-th zero in $B_1$ ($B_2$) if position $i$ ($i-1$) is the 
starting (ending) position of a chunk in $A$. We store these two bit vectors 
using the FID structure of Theorem~\ref{thm:fid},
which takes $O(n \log k / k)$ bits of space.
Now, considering the $1$s in $B_1$ as open parentheses and the 
$1$s in $B_2$ as close parentheses, we `merge' the two parenthesis 
sequences to obtain a balanced parenthesis sequence that 
represents the tree structure of the microtrees. We store this sequence 
along with an auxiliary structure to support parenthesis operations,
such as find-open, find-close, excess, rank-open, rank-close~\cite{parentheses}. 
Using all these data structures, we can support the following operations:
(i) number the microtrees in
the sorted order of their preorder number of the leftmost leaf, 
(ii) find the microtree that contains the leaf corresponding to a given 
column index in $A$, and (iii) find the chunk corresponding to a 
given microtree. 

For each microtree $\mu$, we store the bits of $M$ corresponding to the 
leaves of $\mu$ in the order of their preorder numbers. Given these bits, 
we now show how to reconstruct $\mu$, in time linear in the size of $\mu$.
Suppose $C_{\ell}$ and $C_r$ are the two chunks corresponding to $\mu$
(where $C_r$ is empty if the nodes in $\mu$ correspond to a single chunk).
We first observe that all the elements that lie between the chunks $C_{\ell}$
and $C_r$ in $A$ are strictly smaller than the maximum element in the
last column of $C_{\ell}$ as well as the the maximum element in the first
column of $C_r$. Thus any RMQ query whose one end point lies in 
$C_{\ell}$ and the other end point lies in $C_r$ has its answer in one these 
two chunks and never in between these chunks.

Let $i$ be the first column of $C_{\ell}$ and $j$ be the last column of $C_r$.
We first find $t = RMQ(A, [1] \times [i..j])$ and $b = RMQ(A, [2] \times [i..j])$ which return
the positions of the maximum elements in the range $[i..j]$ in the top and 
bottom rows respectively. The first bit of $M$ that we store for the microtree
is a 0 or 1 depending on whether $A[1,t]$ is greater  or less than $A[2,b]$.
Suppose that $A[1,t]$ is the larger of the two (the other case is similar). 
Then the root of $\mu$ has a leaf child which corresponds to the position 
$t$. The leftmost (first) and rightmost (third) subtrees of the root correspond
to all the positions in the chunks $C_{\ell}$ and $C_r$ that are to the left and
right of the position $t$ respectively, and can be constructed recursively using the same
procedure. Since each node of $\mu$ can be `constructed' in $O(1)$ time, 
the overall time to reconstruct $\mu$ is $O(|\mu|)$, where $|\mu|$ denotes 
the size of $\mu$.

To answer an RMQ query that spans both the rows, we use an algorithm 
that answers a query by storing the tree decomposition representation.
Whenever we need to access a microtree, we reconstruct it using the above 
procedure, in $O(k)$ time; all the auxiliary structures are stored explicitly, and 
hence can be accessed in $O(1)$ time. Since the query algorithm accesses
$O(1)$ microtrees, the overall running time is $O(k)$.
Finally, the lower order terms that arise in various substructures above, such as 
FIDs and auxiliary structures for tree representations, which are independent of 
$k$ can be made $O(n / (\log n)^{O(1)})$ using the ideas from~\cite{Patrascu08}.
Thus the overall space used is $5n + O(n \log k/ k)$ bits, for any parameter 
$k = (\log n)^{O(1)}$.
\qed
\end{proof}



\section{Conclusions and Open Problems}

We have given new explicit encodings for the RMQ problem, as well as (in some cases)
efficient data structures whose space usage is close to the sizes of the encodings.  
We have focused on the cases of random matrices (which may
have relevance to practical applications such as OLAP cubes \cite{olap}) and the case
of small values of $m$.  Obviously, the problem of determining the asymptotic complexity of
encoding RMQ for general $m$ remains open.

\bibliographystyle{abbrv}
\bibliography{rangemax}

\begin{thebibliography}{10}

\bibitem{Amir2007}
A.~Amir, J.~Fischer, and M.~Lewenstein.
\newblock Two-dimensional range minimum queries.
\newblock In {\em Proc. 18th Annual Symposium on Combinatorial Pattern
  Matching}, volume 4580 of {\em LNCS}, pages 286--294. Springer-Verlag, 2007.

\bibitem{Atallah2009a}
M.~J. Atallah and H.~Yuan.
\newblock Data structures for range minimum queries in multidimensional arrays.
\newblock In {\em Proc. 20th Annual ACM-SIAM Symposium on Discrete Algorithms},
  pages 150--160. SIAM, 2010.

\bibitem{BenoitDMRRR05}
D.~Benoit, E.~D. Demaine, J.~I. Munro, R.~Raman, V.~Raman, and S.~S. Rao.
\newblock Representing trees of higher degree.
\newblock {\em Algorithmica}, 43(4):275--292, 2005.

\bibitem{Brodal2010}
G.~S. Brodal, P.~Davoodi, and S.~S. Rao.
\newblock {On Space Efficient Two Dimensional Range Minimum Data Structures}.
\newblock In {\em Proc. of European Symposium on Algorithms}, volume 6347 of
  {\em Lecture Notes in Computer Science}, pages 171--182. Springer, 2010.

\bibitem{olap}
S.~Chaudhuri and U.~Dayal.
\newblock An overview of data warehousing and {OLAP} technology.
\newblock {\em SIGMOD Rec.}, 26:65--74, March 1997.

\bibitem{Demaine2009}
E.~D. Demaine, G.~M. Landau, and O.~Weimann.
\newblock On cartesian trees and range minimum queries.
\newblock In {\em Proc. 36th International Colloquium on Automata, Languages
  and Programming}, volume 5555 of {\em LNCS}, pages 341--353. Springer-Verlag,
  2009.

\bibitem{FM08}
A.~Farzan and J.~I. Munro.
\newblock A uniform approach towards succinct representation of trees.
\newblock In J.~Gudmundsson, editor, {\em SWAT}, volume 5124 of {\em Lecture
  Notes in Computer Science}, pages 173--184. Springer, 2008.

\bibitem{DBLP:journals/jacm/FerraginaM05}
P.~Ferragina and G.~Manzini.
\newblock Indexing compressed text.
\newblock {\em JACM}, 52:552--581, 2005.

\bibitem{DBLP:journals/siamcomp/FischerH11}
J.~Fischer and V.~Heun.
\newblock Space-efficient preprocessing schemes for range minimum queries on
  static arrays.
\newblock {\em SIAM J. Comput.}, 40(2):465--492, 2011.

\bibitem{Gabow1984}
H.~N. Gabow, J.~L. Bentley, and R.~E. Tarjan.
\newblock Scaling and related techniques for geometry problems.
\newblock In {\em Proc. 16th Annual ACM Symposium on Theory of Computing},
  pages 135--143. ACM, 1984.

\bibitem{parentheses}
R.~F. Geary, N.~Rahman, R.~Raman, and V.~Raman.
\newblock A simple optimal representation for balanced parentheses.
\newblock {\em Theor. Comput. Sci.}, 368(3):231--246, 2006.

\bibitem{DBLP:journals/siamcomp/GrossiV05}
R.~Grossi and J.~S. Vitter.
\newblock Compressed suffix arrays and suffix trees with applications to text
  indexing and string matching.
\newblock {\em SICOMP}, 35(2):378--407, 2005.

\bibitem{DBLP:reference/algo/HowardV08}
P.~G. Howard and J.~S. Vitter.
\newblock Arithmetic coding for data compression.
\newblock In M.-Y. Kao, editor, {\em Encyclopedia of Algorithms}. Springer,
  2008.

\bibitem{Jacobson89}
G.~Jacobson.
\newblock Space-efficient static trees and graphs.
\newblock In {\em FOCS}, pages 549--554. IEEE, 1989.

\bibitem{Kieffer2009}
J.~C. Kieffer, E.-H. Yang, and W.~Szpankowski.
\newblock Structural complexity of random binary trees.
\newblock In {\em Proc. IEEE International Symposium on Information Theory
  (ISIT)}, pages 635--639, 2009.

\bibitem{Patrascu08}
M.~Patrascu.
\newblock Succincter.
\newblock In {\em FOCS}, pages 305--313. IEEE Computer Society, 2008.

\bibitem{DBLP:conf/soda/RamanRR02}
R.~Raman, V.~Raman, and S.~S. Rao.
\newblock Succinct indexable dictionaries with applications to encoding k-ary
  trees and multisets.
\newblock In {\em SODA}, pages 233--242, 2002.

\bibitem{Sadakane02-ACT}
K.~Sadakane.
\newblock Space-efficient data structures for flexible text retrieval systems.
\newblock In P.~Bose and P.~Morin, editors, {\em ISAAC}, volume 2518 of {\em
  Lecture Notes in Computer Science}, pages 14--24. Springer, 2002.

\bibitem{Sadakane2007}
K.~Sadakane.
\newblock Compressed suffix trees with full functionality.
\newblock {\em Theory of Computing Systems}, 41(4):589--607, 2007.

\bibitem{SN10}
K.~Sadakane and G.~Navarro.
\newblock Fully-functional succinct trees.
\newblock In M.~Charikar, editor, {\em SODA}, pages 134--149. SIAM, 2010.

\bibitem{DBLP:journals/siamcomp/SchmidtS90}
J.~P. Schmidt and A.~Siegel.
\newblock The spatial complexity of oblivious k-probe hash functions.
\newblock {\em SIAM J. Comput.}, 19(5):775--786, 1990.

\bibitem{Vuillemin1980}
J.~Vuillemin.
\newblock A unifying look at data structures.
\newblock {\em Communications of the ACM}, 23(4):229--239, 1980.

\end{thebibliography}
\end{document}